%% file: main.tex
\documentclass[numbers]{sigplanconf}
\usepackage[utf8]{inputenc}
\usepackage{color,xspace}
\usepackage{amsmath}
\usepackage{amssymb}
\usepackage{amsthm}
\usepackage{stmaryrd}
\usepackage{flushend}

\usepackage{url}                  
\usepackage{listings}          
\usepackage{enumerate}
\usepackage{paralist}
\usepackage{caption}
\usepackage{subcaption}
\usepackage{multirow}
\usepackage{graphicx}

\usepackage{paralist}

\usepackage{float}

\toappear{}

\input{macros}

\title{Rigorous Analysis of Software Countermeasures against Cache
  Attacks }

\authorinfo{Goran Doychev \and Boris K{\"o}pf}
           {IMDEA Software Institute}
           {\{goran.doychev,boris.koepf\}@imdea.org}

\begin{document}

\maketitle

\begin{abstract}
  CPU caches introduce variations into the execution time of programs
  that can be exploited by adversaries to recover private information
  about users or cryptographic keys.

  Establishing the security of countermeasures against this threat
  often requires intricate reasoning about the interactions of program
  code, memory layout, and hardware architecture and has so far only
  been done for restricted cases.

  In this paper we devise novel techniques that provide support for
  bit-level and arithmetic reasoning about memory accesses in the presence of
  dynamic memory allocation. These techniques
  enable us to perform the
  first rigorous analysis of widely deployed software countermeasures
  against cache attacks on modular exponentiation, based on executable
  code.
\end{abstract}

%
%

\input{intro}

\input{example}
\input{security}

\input{quantification}

\input{obsdomains}

\input{soundness}

\input{experiments}

\input{related}

\section{Conclusions}\label{sec:conclusions}
In this paper we devise novel techniques that provide support for
bit-level and arithmetic reasoning about pointers in the presence of
dynamic memory allocation.  These techniques enable us to perform the
first rigorous analysis of widely deployed software countermeasures
against cache side-channel attacks on modular exponentiation, based on
executable code.

\paragraph{Acknowledgments} We thank Roberto Giacobazzi, Francesco
Logozzo, Laurent Mauborgne, and the anonymous reviewers for their
helpful feedback. This work was supported by Ram{\'o}n y Cajal grant
RYC-2014-16766, Spanish projects TIN2012-39391-C04-01 StrongSoft and
TIN2015-70713-R DEDETIS, and Madrid regional project S2013/ICE-2731
N-GREENS.

\bibliographystyle{abbrv}
\bibliography{leakmc}

%

\end{document}

%% file: macros.tex
\definecolor{darkred}{rgb}{0.6,0,0}

\newcommand{\gorancomment}[1]{}
\newcommand{\boriscomment}[1]{}
\newcommand{\TODO}[1]{}

\newcommand{\cnt}{\mathit{cnt}^\projgen}

\newcommand{\prog}{P\xspace} 
\newcommand{\Trans}{\ensuremath{\mathcal{R}}} 
\newcommand{\Traces}{{\mathit{Traces}}}
\newcommand{\nextop}{\ensuremath{\mathop{\mathit{next}}\nolimits}}
\newcommand{\absnext}{\ensuremath{\mathop{\abs{\mathit{next}}}\nolimits}}
\newcommand{\States}{\ensuremath{\Sigma}} 
\newcommand{\StatesInit}{\ensuremath{I}} 

\newcommand{\Hi}{\mathit{hi}}
\newcommand{\StatesInitHi}{\StatesInit_\Hi}
\newcommand{\Lo}{\mathit{lo}}
\newcommand{\StatesInitLo}{\StatesInit_\Lo}

\newcommand{\Trace}{\ensuremath{{\mathcal T}}}
\newcommand{\MSym}{\ensuremath{{\mathcal M}}}

\newcommand{\state}{\sigma} 
\newcommand{\statelo}{\lambda}
\newcommand{\stateloext}{{\bar{\statelo}}}

\newcommand{\collSemLo}{\collSem_{\statelo}}

\newcommand{\maskop}{\odot}

\newcommand{\symbmapping}{\lambda}
\newcommand{\loSym}{\Sym_\Lo}

\newcommand{\Addresses}{\mathcal{A}}
\newcommand{\blocktrace}{\text{\em block}}
\newcommand{\addrtrace}{\text{\em address}}
\newcommand{\weaktrace}{\text{\em b-block}}
\newcommand{\banktrace}{\text{\em bank}}
\newcommand{\pagetrace}{\text{\em page}}

\newcommand{\view}{\mathit{view}}
\newcommand{\viewaddr}{\mathit{view}^\addrtrace}
\newcommand{\viewblock}{\mathit{view}^\blocktrace}
\newcommand{\viewweakblock}{\mathit{view}^\weaktrace}
\newcommand{\viewbank}{\mathit{view}^\banktrace}
\newcommand{\viewpage}{\mathit{view}^\pagetrace}

\newcommand{\Events}{\ensuremath{\Addresses^*}}
\newcommand{\evt}{a} 

\newcommand{\upd}{\mathit{upd}}

\newcommand\abs[1]{{#1}^\sharp}

\newcommand{\AbsMS}{\abs{\MSym}}
\newcommand{\AbsTR}{\abs{\Trace}}

\newcommand{\conc}{\gamma}
\newcommand{\concTR}{\conc^{\AbsTR}}
\newcommand{\concMS}{\conc^{\AbsMS}}

\newcommand{\collSem}{\mathit{Col}} 
\newcommand{\Abscoll}{\abs{\collSem}}

\newcommand{\sizeof}[1]{|#1|} 

\newcommand\range[1]{\mathop{ran}(#1)}
\newcommand\domain[1]{\mathop{dom}(#1)}
\newcommand{\partsof}[1]{\mathcal{P}(#1)}    

\newcommand{\bitlen}{n}

\newcommand{\constsym}{s}
\newcommand{\mask}{m}

\newcommand{\Sym}{\mathit{Sym}}
\newcommand{\Bitvecs}{\{0,1\}^\bitlen}
\newcommand{\Masks}{\{0,1,\top\}^\bitlen}
\newcommand{\absms}{\abs{\ms}}
\newcommand{\abstr}{\abs{\tr}}
\newcommand{\absabs}{\abs{a}}
\newcommand{\ms}{x}
\newcommand{\tr}{t}

\newcommand{\Leakage}{\mathcal{L}}

\newcommand{\logesize}{b}

\newcommand{\compsteps}{k}

\newcommand{\projgen}{\pi}
\newcommand{\proj}{\projgen_{\bitlen:\logesize}}

\newcommand{\Ext}{\mathit{Ext}}
\newcommand{\Labl}{L}
\newcommand{\Rept}{R}
\newcommand{\Nodes}{V}
\newcommand{\node}{v}
\newcommand{\parent}{u}
\newcommand{\Edges}{E}
\newcommand{\rt}{r}

\newcommand{\origin}{\ensuremath{\mathit{orig}}}
\newcommand{\suck}{\ensuremath{\mathit{succ}}}
\newcommand{\offset}{\ensuremath{\mathit{off}}}
\newcommand{\intoff}{c}
\newcommand{\nosuccessor}{\bot}

\newcommand{\ttop}{\mathop{\texttt{OP}}}
\newcommand{\ttxor}{\mathop{\texttt{XOR}}}
\newcommand{\ttor}{\mathop{\texttt{OR}}}
\newcommand{\ttand}{\mathop{\texttt{AND}}}
\newcommand{\ttadd}{\mathop{\texttt{ADD}}}
\newcommand{\ttcmp}{\mathop{\texttt{CMP}}}
\newcommand{\ttsub}{\mathop{\texttt{SUB}}}

\newcommand{\absop}{\mathop{\abs{\texttt{OP}}}}

\newcommand{\absxor}{\mathop{\abs{\texttt{XOR}}}}
\newcommand{\absadd}{\mathop{\abs{\texttt{ADD}}}}
\newcommand{\abssub}{\mathop{\abs{\texttt{SUB}}}}
\newcommand{\abscmp}{\mathop{\abs{\texttt{CMP}}}}

\newtheorem{example}{Example}
\newtheorem{proposition}{Proposition}
\newtheorem{theorem}{Theorem}
\newtheorem{lemma}{Lemma}

\newif\ifext
\exttrue

%% file: intro.tex
\section{Introduction}
CPU caches reduce the latency of memory accesses on average, but not
in the worst case. Thus, they introduce variations into the execution
time of programs, which can be exploited by adversaries to recover secrets, such as
private information about users or cryptographic
keys~\cite{Bernstein05cache-timingattacks,osvikshamir06cache,AciicmezSK07,RistenpartTSS09,
	GullaschBK11,YaromF14}.

A large number of techniques have been proposed to counter this
threat. Some proposals work at the level of the operating
system~\cite{erlingsson2007operating,KimP12,zhang2013duppel}, others
at the level of the hardware
architecture~\cite{intelaes10,WangL08,TiwariOLVLHKCS11} or the
cryptographic protocol~\cite{DzPi_08}. In practice, however, software
countermeasures are often the preferred choice because they can be
easily deployed.

A common software countermeasure is to ensure that control
flow, memory accesses, and execution time of individual instructions
do not depend on secret
data~\cite{langley-valgrind,BernsteinLS12}. While such code prevents
 leaks through instruction and data caches, hiding all
dependencies can come with performance penalties~\cite{CoppensVBS09}.

More permissive countermeasures are to ensure that both branches of
each conditional fit into a single line of the instruction cache, to
preload lookup tables, or to permit secret-dependent memory access
patterns as long as they are secret-{\em in}dependent at the
granularity of cache lines or sets. Such permissive code can be faster
and is widely deployed in crypto-libraries such as OpenSSL. However,
analyzing its security requires intricate reasoning about the
interactions of the program and the hardware platform and has so far
only been done for restricted cases~\cite{cacheaudit-tissec15}.

A major hurdle for reasoning about these interactions are the
requirements put on tracking memory addresses:
On the one hand, static analysis of code with dynamic memory
allocation requires memory addresses to be dealt with symbolically.
On the other hand, analysis of cache-aligned memory layout
requires support for accurately tracking the effect of bit-level and arithmetic
operations. 
While there are solutions that address each of these requirements in
isolation, supporting them together is challenging, because the demand
for symbolic treatment conflicts with the demand for bit-level
precision. 

In this paper, we propose novel techniques that meet both requirements
and thus enable the automated security analysis of permissive software
countermeasures against microarchitectural side-channel attacks in
executable code.

\paragraph{Abstract Domains} Specifically, we introduce {\em masked
	symbols}, which are expressions that represent unknown addresses,
together with information about some of their bits. Masked symbols
encompass unknown addresses as well as known constants; more
importantly, they also support intermediate cases, such as addresses
that are unknown except for their least significant bits, which are
zeroed out to align with cache line boundaries.  We cast arithmetic
and bit-level operations on masked symbols in terms of a simple
set-based abstract domain, which is a data structure that supports
approximating the semantics of programs~\cite{cousot:cousot77}. We
moreover introduce a DAG-based abstract domain to represent sets of
traces of masked symbols.

\paragraph{Adversary Models}
Our novel abstract domains enable us to reason about the security of
programs against a wide range of microarchitectural side channel
adversaries, most of which are out of the scope of existing
approaches. The key observation is that the capability of these
adversaries to observe a victim's execution can be captured in terms
of projections to some of the bits of the addresses accessed by the
victim. This modeling encompasses adversaries that can see the full
trace of accesses to the instruction cache (commonly known as the {\em
	program counter security model}~\cite{molnar2006program}), but also
weaker ones that can see only the trace of memory pages, blocks, or
cache banks, with respect to data, instruction, or shared caches.

\paragraph{Bounding Leaks}
We use our abstract domains for deriving upper bounds on the amount of
information that a program leaks. We achieve this by counting the
number of observations each of these adversaries can make during
program execution, as in~\cite{lowe02quant,newsome09,koepfrybal10}. 
In this paper we
perform the counting by applying an adversary-specific projection to
the set of masked symbols corresponding to each memory access. We
highlight two features of this approach:
\begin{asparaitem}
	\item The projection may collapse a
	multi-element set to a singleton set, for example, in the case of
	different addresses mapping to the same memory block. This is the key
	for establishing that some memory accesses do not leak information to
	some observers, even if they depend on secret data.
	\item As the projection operates on individual bits, we can compute the
	adversary's observations on addresses that contain both known and
	unknown bits. In this way, our counting effectively captures the
	leak of the program, rather than the uncertainty about the address
	of the dynamically allocated memory.
\end{asparaitem}

\paragraph{Implementation and Evaluation}

We implement our novel techniques on top of the CacheAudit static binary
analyzer~\cite{cacheaudit-tissec15}, and we evaluate their
effectiveness in a case study where we perform the first formal
analysis of commonly used software countermeasures for protecting
modular exponentiation algorithms. The paper contains a detailed
description of our case study; here we highlight the following
results:
\begin{asparaitem}
	\item We analyze the security of the scatter/gather countermeasure
	used in OpenSSL~1.0.2f for protecting window-based modular
	exponentiation. Scatter/gather ensures that the pattern of data
	cache accesses is secret-independent at the level of granularity of
	cache lines and, indeed, our analysis of the binary executable
	reports security against adversaries that can monitor only cache
	line accesses.
	\item Our analysis of the scatter/gather countermeasure reports a
	leak with respect to adversaries that can monitor memory accesses at
	a more fine-grained resolution. This leak has been
	exploited in the CacheBleed attack~\cite{yarom16cachebleed}, where
	the adversary observes accesses to the individual banks within a
	cache line.  We analyze the variant of scatter/gather published in
	OpenSSL 1.0.2g as a response to the attack and prove its security
	with respect to powerful adversaries that can monitor the full
	address trace.
	\item Our analysis detects the side channel in the square-and-multiply
	 algorithm in libgcrypt~1.5.2 that has been exploited
	in~\cite{YaromF14,LiuYGHL15}, but can prove the absence of an
	instruction cache leak in the square-and-{\em
		always}-multiply algorithm used in libgcrypt~1.5.3, for some compiler 
	optimization levels.
\end{asparaitem}

Overall, our results illustrate (once again) the dependency of
software countermeasures against cache attacks on brittle details of
the compilation and the hardware architecture, and they demonstrate
(for the first time) how automated program analysis can effectively
support the rigorous analysis of permissive software countermeasures.

In summary, our contributions are to devise novel techniques that
enable cache-aware reasoning about dynamically allocated memory, and
to put these techniques to work in the first rigorous analysis of
widely deployed permissive countermeasures against cache side channel
attacks.

%
%



%% file: example.tex
\section{Illustrative Example}\label{sec:illustration}

We illustrate the scope of the techniques developed in this paper
using a problem that arises in implementations of windowed modular
exponentiation. There, powers of the base are pre-computed and stored
in a table for future lookup. Figure~\ref{fig:precomp-layout} shows an
example memory layout of two such pre-computed values $p_2$ and $p_3$,
each of $3072$ bits, which are stored in heap memory. An adversary that 
observes accesses to the six
memory blocks starting at \texttt{80eb140} knows that $p_2$ was
requested, which can give rise to effective key-recovery
attacks~\cite{LiuYGHL15}.

\begin{figure}[h]
\centering
\includegraphics[width=.4\textwidth]{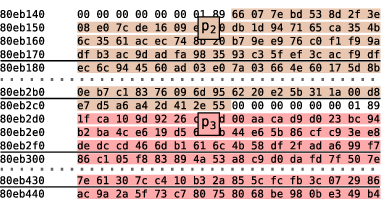}
\caption{Layout of pre-computed values in main memory, for the
  windowed modular exponentiation implementation from
  libgcrypt~1.6.1. Black lines denote the memory block boundaries, for
  an architecture with blocks of 64 bytes.}\label{fig:precomp-layout}
\end{figure}

Defensive approaches for table lookup, as implemented in NaCl or
libgcrypt 1.6.3, avoid such vulnerabilities by accessing \emph{all}
table entries in a constant order.  OpenSSL~1.0.2f instead uses a
more permissive approach that accesses only {\em one} table entry,
however it uses a smart layout of the tables to ensure that the
memory blocks are loaded into the cache in a constant order.
An example layout for storing $8$ pre-computed values is shown in
Figure~\ref{fig:scatter-gather-layout}.
\begin{figure}[h]
\centering
\includegraphics[width=.41\textwidth]{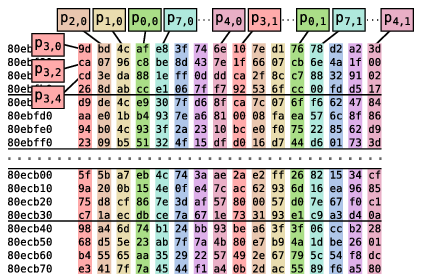}
\caption{Layout of pre-computed values in main memory, achieved with
  the scatter/gather countermeasure. Data highlighted in different
  colors correspond to pre-computed values $p_0,\dots,p_7$,
  respectively.  Black lines denote the memory block boundaries, for
  an architecture with blocks of 64
  bytes.}\label{fig:scatter-gather-layout}
\end{figure}
The code that manages such tables consists of three functions,
which are given in Figure~\ref{fig:CM-openssl}.
\begin{asparaitem}
\item To create the layout, the function \texttt{align} aligns a
  buffer with the memory block boundary by ensuring the
  least-significant bits of the buffer address are zeroed.
\item To write a value into the array, the function \texttt{scatter} ensures
that the bytes of the pre-computed values are stored \texttt{spacing}
bytes apart. 
\item Finally, to retrieve a pre-computed value from the buffer, the
  function~\texttt{gather} assembles the value by accessing its bytes
  in the same order they were stored.
\end{asparaitem}

\begin{figure}[h]
 \centering
\begin{lstlisting}[mathescape]
align ( buf ):
    return buf - ( buf & ( block_size - 1 ) ) + block_size

scatter ( buf, p, k ):
    for i := 0 to N - 1 do
        buf [k + i * spacing] := p [k][i]

gather ( r, buf, k ):
    for i := 0 to N - 1 do
        r [i] := buf [k + i * spacing]
\end{lstlisting}
 

\caption{Scatter/gather method from OpenSSL~1.0.2f for aligning, storing 
and retrieving pre-computed values.}\label{fig:CM-openssl}
\end{figure}


The techniques developed in this paper enable 
automatic reasoning about the effectiveness of such countermeasures,
for a variety of adversaries. Our analysis handles the
dynamically-allocated address in \texttt{buf} from
Figure~\ref{fig:CM-openssl} symbolically, but is still able to
establish the effect of \texttt{align} by
considering bit-level semantics of arithmetic operators on symbols:
First, the analysis establishes that \texttt{buf \& (block\_size - 1)}
clears the most-significant bits of the symbol $s$; second, when
subtracting this value from \texttt{buf}, the analysis determines that
the result is $s$, with the least-significant bits cleared; third, the
analysis determines that adding \texttt{block\_size} leaves the
least-significant bits unchanged, but affects the unknown bits,
resulting in a new symbolic address $s'\neq s$ whose least significant bits are cleared.

Using this information, in \texttt{gather}, our 
analysis establishes that, independently from the value of $k$, at each 
iteration of the loop, the most-significant bits of the accessed location are 
the same. Combining this information with knowledge about the  
architecture such as the block size, the analysis establishes that the sequence 
of accessed memory blocks is the same, thus the countermeasure ensures
security of scatter/gather with respect to
adversary who makes observations at memory-block granularity.

%% file: security.tex
\section{Memory Trace Adversaries}\label{sec:security}

In this section we formally define a family of side channel
adversaries that exploit features of the microarchitecture, in
particular: caches.  The difference between these adversaries is the
granularity at which they can observe the trace of programs'
accesses to main memory. We start by introducing an abstract notion
of programs and traces.

\subsection{Programs and Traces}
We introduce an abstract notion of programs as the transformations of
the main memory and CPU register contents (which we collectively call the 
\emph{machine state}), caused by the execution of the program's 
instructions. Formally, a program
$\prog=(\States,\StatesInit,\Addresses,\Trans)$
consists of the following components:
\begin{compactitem}
	\item $\States$ - a set of {\em states}
	\item $\StatesInit\subseteq\States$ - a set of possible {\em initial} states
	\item $\Addresses$ - a set of {\em addresses}
	\item $\Trans \subseteq \States \times \Events \times \States$ - a
	{\em transition relation}
\end{compactitem}

A transition $(\state_i, \evt ,\state_j)\in\Trans$ captures two
aspects of a computation step: first, it describes how the instruction
set semantics operates on data stored in the machine state, namely by
updating $\state_i$ to $\state_j$; second, it describes the sequence
of memory accesses $\evt\in\Events$ issued during this update, which
includes the addresses accessed when fetching instructions from the
code segment, as well as the addresses containing accessed data. To
capture the effect of one computation step in presence of uncertain
inputs, we define the $\nextop{}$ operator:
\begin{equation*}
\nextop(S)=\{t.\state_\compsteps\evt_\compsteps\state_{\compsteps+1}\mid 
t.\state_\compsteps\in S \wedge
(\state_\compsteps,\evt_\compsteps,\state_{\compsteps+1})\in\Trans \}\ .
\end{equation*}
A {\em trace} of $\prog$ is an alternating sequence of states
and addresses $\state_0\evt_0\state_1\evt_1\ldots\state_\compsteps$ such that
$\state_0\in\StatesInit$, and that for all $i\in\{0,\dots,\compsteps-1\}$,
$(\state_i,\evt_i,\state_{i+1})\in\Trans$. The set of all traces
of $\prog$ is its {\em collecting semantics} 
$\collSem\subseteq\Traces$.
In this paper, we only consider terminating programs, and define their
collecting semantics  as the least fixpoint of the $\nextop$
operator containing~\StatesInit: 
$
\collSem=\StatesInit\cup\nextop(\StatesInit)\cup\nextop^2(\StatesInit)
\cup\dots\ .
$

\subsection{A Hierarchy of Memory Trace
	Observers}\label{sec:hierarchy}

Today's CPUs commonly partition the memory space into units of
different sizes, corresponding to virtual memory pages, cache lines,
or cache banks. The delays introduced by page faults, cache misses, or
bank conflicts enable real-world adversaries to effectively identify
the units involved in another program's memory accesses.  We
explicitly model this capability by defining adversaries that can
observe memory accesses at the granularity of each unit, but that
cannot distinguish between accesses to positions within the same unit.

\paragraph{Observing a Memory Access} 
On a common $\bitlen$-bit architecture, the most significant
$\bitlen-\logesize$ bits of each address serve as an identifier for
the unit containing the addressed data, and the least significant
$\logesize$ bits serve as the offset of the data within that
unit, where $2^\logesize$ is the byte-size of the respective unit.

We formally capture the capability to observe units of size $2^\logesize$
by projecting addresses to their $\bitlen-\logesize$ most significant
bits, effectively making the $\logesize$ least significant bits
invisible to the adversary. That is, when accessing the $\bitlen$-bit
address $a=(x_{\bitlen-1},x_{\bitlen-2},\dots,x_0)$, the adversary
observes
\begin{equation*}
\proj(a) := (x_{\bitlen-1},x_{\bitlen-2},\dots,x_{\logesize})\ .
\end{equation*}

\begin{example}
	A 32-bit architecture with 4KB pages, 64B cache lines, and 4B cache
	banks will use bits 0 to 11 for offsets within a page, 0 to 5 for
	offsets within a cache line, and 0 to 1 for offsets within a cache
	bank. That is, the corresponding adversaries observe bits 12 to 31,
	6 to 31, and 2 to 31, respectively, of each memory access.
\end{example}

\paragraph{Observing Program Executions}
We now lift the capability to observe individual memory accesses to full
program executions. This lifting is formalized in terms of {\em
	views}, which are functions that map traces in $\collSem$ to
sequences of projections of memory accesses to observable
units. Formally, the view of an adversary on a trace of the program is
defined by
\begin{equation*}
\view\colon
\state_0\evt_0\state_1\evt_1\dots \state_\compsteps \mapsto
\proj(\evt_0)\proj(\evt_1)\dots\proj(\evt_{\compsteps-1})\ .
\end{equation*}
By considering $\proj$ for different values of $\logesize$, we obtain a
hierarchy of memory trace observers:
\begin{asparaitem}
	\item The {\em address-trace observer} corresponds to $\logesize=0$; it can
	observe the full sequence $\evt_0\evt_1\dots\evt_{\compsteps-1}$ of
	memory locations that are accessed.  Security against this adversary
	implies resilience to many kinds of microarchitectural side
	channels, through cache, TLB, DRAM, and branch prediction
	buffer.\footnote{We do not model, or make assertions about, the
		influence of advanced features such as out-of-order-execution.}
	An address-trace observer restricted to instruction-addresses is
	equivalent to the program counter security
	model~\cite{molnar2006program}.
	\item The {\em block-trace observer} can observe the sequence of
	memory blocks loaded into cache lines. Blocks
	are commonly of size $32$, $64$, or $128$ bytes, i.e. $\logesize=5$,
	$6$, or $7$. Security against this
	adversary implies resilience against adversaries that can monitor
	memory accesses at the level of granularity of cache lines.  Most known
	cache-based attacks exploit observations at the
	granularity of cache lines,
	e.g.~\cite{Percival05cachemissing,ZhangJRR12,LiuYGHL15}.
	\item The {\em bank-trace observer} can observe a sequence of accessed
	cache banks, a technology used in some CPUs for hyperthreading. 
	An example of an attack at the granularity of cache
	banks is CacheBleed~\cite{yarom16cachebleed} against the
	scatter/gather implementation from OpenSSL 1.0.2f. The platform
	targeted in this attack has 16 banks of size $4$ bytes, i.e. $\logesize=2$.
	\item The {\em page-trace observer} can observe memory accesses at the
	granularity of accessed memory pages, which are commonly of size
	$4096$ bytes, i.e. $\logesize=12$. Examples of such attacks appear
	in~\cite{xu2015controlled} and~\cite{shinde06preventing}.
\end{asparaitem}
We denote the views of these observers by $\viewaddr$,
$\viewblock$, $\viewbank$, and $\viewpage$, respectively.

\paragraph{Observations Modulo Stuttering}
\gorancomment{still not good}
For each of the observers defined above we also consider a variant
that cannot distinguish between repeated accesses to the same unit
 (which we call {\em stuttering}). This is motivated
 by the fact that the latency of cache misses dwarfs that of cache hits
 and is hence easier to observe.

For the observer $\viewblock$, we formalize this variant in terms of a
function $\viewweakblock$ taking as input a block-sequence $w$ 
and replacing the maximal subsequences $B\cdots B$ of each block $B$ in $w$
by the single block $B$. E.g., $\viewweakblock$ maps both
$AABCDDC$ and $ABBBCCDDCC$ to the sequence $ABCDC$, making them
indistinguishable to the adversary.
This captures an adversary
that cannot count the number of memory accesses, as long as they
are guaranteed to be cache~hits\footnote{Here we rely on the
	(weak) assumption that the second $B$ in any access sequence
	$\cdots BB\cdots$ is guaranteed to hit the cache.}.

%

%% file: quantification.tex
\section{Static Quantification of Leaks}\label{sec:quantleak}
In this section, we characterize the amount of information leaked by a
program, and we show how this amount can be over-approximated by
static analysis. While the basic idea is standard (we rely on upper
bounds on the cardinality of an adversary's view), our presentation
exhibits a new path for performing such an analysis in the presence of
low inputs. In this section we outline the basic idea, which we
instantiate in Sections~\ref{sec:statedom} and~\ref{sec:tracedom} for
the observers defined in Section~\ref{sec:security}.

\paragraph{Quantifying Leaks}
As is common in information-flow analysis, we quantify the amount of
information leaked by a program about its secret (or {\em high})
inputs in terms of the maximum number of observations an adversary can
make, for any valuation of the public (or {\em low})
input~\cite{lowe02quant,koepfbasin07,smith09}.

To reflect the distinction between high and low inputs in the
semantics, we split the initial state into a low part $\StatesInitLo$
and a high part $\StatesInitHi$, i.e.,
$\StatesInit = \StatesInitLo\times\StatesInitHi$. We split the
collecting semantics into a family of collecting semantics
$\collSemLo$ with $I = \{\statelo\}\times\StatesInitHi$, one for each
$\statelo\in\StatesInitLo$, such that
$\collSem = \bigcup\limits_{\statelo}\collSemLo$.

 Formally, we define \emph{leakage} as
 the maximum cardinality of the adversary's view w.r.t. all
possible low inputs:
\begin{equation}\label{eq:leakage-low}
  \Leakage := \max_{\statelo\in\StatesInitLo}(
  \sizeof{\view(\collSemLo)})\ .
\end{equation}
The logarithm of this number corresponds to the number of leaked {\em
  bits}, and it comes with different interpretations in terms of
security. For example, it can be related to a lower bound on the
expected number of guesses an adversary has to make for successfully
recovering the secret input~\cite{massey94}, or to an upper bound on
the probability of successfully guessing the secret input in one
shot~\cite{smith09}.  Note that a leakage $\Leakage$ of $1$ (i.e. $0$
bits) corresponds to non-interference.

\paragraph{Static Bounds on Leaks}

For quantifying leakage based on Equation~\ref{eq:leakage-low}, one
needs to determine the size of the range of $\view$ applied to the
fixpoint $\collSemLo$ of the $\nextop$ operator, for all
$\statelo\in\StatesInitLo$ -- which is infeasible for most programs.

For {\em fixed} values $\statelo\in \StatesInitLo$, however, the
fixpoint computation can be tractably approximated by abstract
interpretation~\cite{cousot:cousot77}. The result is a so-called
abstract fixpoint $\abs{\collSem}$ that represents a superset of
$\collSemLo$, based on which one can over-approximate the range of
$\view$~\cite{koepfrybal10}.  One possibility to obtain bounds for the
leakage that hold for {\em all} low values is to compute one fixpoint
w.r.t. all possible $\StatesInitLo$ rather than one for each single
$\statelo\in\StatesInitLo$~\cite{cacheaudit-tissec15}. The problem
with this approach is that possible variation in low inputs is
reflected in the leakage, which can lead to imprecision.  
\paragraph{Secret vs Public, Known vs Unknown Inputs}

The key to our treatment of low inputs is that we classify variables
along two dimensions.
\begin{asparaitem}
\item The first dimension is whether variables carry secret (or {\em
    high}) or public (or {\em low}) data. High variables are
  represented in terms of the set of all possible values the variables can
  take, where larger sets represent more uncertainty about the values
  of the variable. Low data is represented in terms of a singleton
  set.
\item The second dimension is whether variables represent values that are
  known at analysis time or not. Known values are represented by
  constants whereas unknown values are represented as symbols.
\end{asparaitem}

\begin{example}
  The set $\{1,2\}$ represents a high variable that carries one of two
  known values. The set $\{s\}$ represents a low variable that carries
  a value $s$ that is not known at analysis time. The set $\{1\}$
  represents a low variable with known value $1$. Combinations such as
  $\{1,s\}$ are possible; this example represents a high variable,
  one of its possible values is unknown at analysis time.
\end{example}
While existing quantitative information-flow analyses that consider
low inputs rely on explicit tracking of path
relations~\cite{backeskoepfrybal09,malacaria16}, our modeling allows
us to identify -- and factor out -- variations in observable outputs
due to low inputs even in simple, set-based abstractions. This enables
us to compute fixpoints $\abs{\collSem}(s)$ containing symbols, based
on techniques that are known to work on intricate low-level code, such
as cryptographic libraries~\cite{cacheaudit-tissec15}. The following
example illustrates this advantage.
\begin{example}\label{ex:simple-if}
  Consider the following program, where variable $x$ is initially
  assigned a pointer to a dynamically allocated memory region. We
  assume that the pointer is low but unknown, which we model by
  $x=\{s\}$, for some symbol $s$. Depending on a secret bit
  $h\in\{0,1\}$ this pointer is increased by 64 or not.

\begin{center}
\begin{lstlisting}[mathescape]
   x:= malloc(1000);
   if h then
       x := x+64 @\label{line:in-if}@
\end{lstlisting}
\end{center}

For an observer who can see the value of $x$ after termination, our
analysis will determine that leakage is bounded by
$\Leakage \le \sizeof{\{s, s+64\}} = 2$. This bound holds for
{\em any} value that $s$ may take in the initial state $\statelo$,
effectively separating uncertainty about low inputs from uncertainty
about high inputs.
\end{example}
In this paper we use low input to model dynamically allocated memory
locations, as in Example~\ref{ex:simple-if}. That is, we
rely on the assumption that locations chosen by the allocator do not
depend on secret data.  More precisely, we assume that the initial
state contains a pool of low but unknown heap locations that can be
dynamically requested by the program.

%% file: obsdomains.tex
\section{Masked Symbol Abstract Domain}\label{sec:statedom}
Cache-aware code often uses Boolean and arithmetic operations on
pointers in order to achieve favorable memory alignment. In this
section we devise the \emph{masked symbol domain}, which is an
abstract domain that enables the static analysis of such code in the
presence of dynamically allocated memory, i.e., when the base pointer
is unknown.


\subsection{Representation}
The masked symbol domain is based on finite sets of what we call {\em
	masked symbols}, which are pairs $(\constsym,\mask)$ consisting of
the following components:
\begin{enumerate}
	\item a {\em symbol} $\constsym\in \Sym$, uniquely identifying an
	unknown value, such as a base address;
	\item a {\em mask} $\mask \in\Masks$, representing a
	pattern of known and unknown bits. We abbreviate the mask
	$(\top,\dots,\top)$ by $\top$.
\end{enumerate}
The $i$-th bit of a masked symbol $(\constsym,\mask)$ is called {\em
	masked} if $\mask_i\in\{0,1\}$, and {\em symbolic} if
$\mask_i=\top$.  Masked bits are known at analysis time, whereas
symbolic bits are not.
Two special cases are worth pointing out: The masked symbol
$(\constsym,\top)$, with $\top$ as shorthand for $(\top,\dots,\top)$,
represents a vector of unknown bits, and $(\constsym,\mask)$ with
$\mask\in\Bitvecs$ represents the bit-vector $\mask$.  In that
way, masked symbols generalize both unknown values and bitvectors.

We use finite sets of masked symbols to represent the elements of the
masked symbol domain, that is, $\AbsMS=\partsof{\Sym\times\Masks}$.

\subsection{Concretization}\label{subsec:masked-concr}
We now give a semantics to elements of the masked symbol domain. This
semantics is parametrized w.r.t.\ valuations of the symbols. For
the case where masked symbols represent partially known heap
addresses, a valuation corresponds to one specific layout of the heap.

Technically, we define the concretization of elements
$\absms\in \AbsMS$ w.r.t.\ a function
$\symbmapping\colon \Sym\rightarrow \Bitvecs$ that maps symbols
to bit-vectors:
\begin{equation*}
\concMS_\symbmapping(\absms)=\{\symbmapping(\constsym) \maskop \mask
\mid
(\constsym,\mask) \in \absms\}
\end{equation*}
Here $\maskop$ is defined bitwise by
$(\symbmapping(\constsym)\maskop m)_i=m_i$ if $m_i\in\{0,1\}$, and
$\symbmapping(\constsym)_i$ if $m_i=\top$.

The function $\symbmapping$ takes the role of the low initial state, for
which we did not assume any specific structure in
Section~\ref{sec:quantleak}. Modeling $\symbmapping$ as a mapping from
symbols to bitvectors is a natural refinement to an initial state
consisting of multiple components that are represented by different
symbols.

\subsection{Counting}
We now show that the precise valuation of the symbols can be ignored
for deriving upper bounds on number of observations that an adversary
can make about a set of masked symbols. 
For this we conveniently interpret a symbol $s$ as a vector of
identical symbols $(s,\dots,s)$, one per bit.\footnote{We use this
	interpretation to track the provenance of bits in projections; it
	does {\em not} imply that the bits of $\symbmapping(s)$ are
	identical.} This allows us to apply the adversary's $\view$ (see
Section~\ref{sec:security}) on masked symbols as the respective
projection $\projgen$ to a subset of observable masked bits.  

Given a
set of masked symbols, we count the observations with respect to the
adversary by applying $\projgen$ on the set's elements and taking the
cardinality of the resulting set.

\begin{example}\label{ex:projectandcount}
	The projection of the set 
	\begin{equation*} 
	\absms=\{(s,(0,0,1)), (t,(\top,\top,1)), (u,(1,1,1))\}
	\end{equation*}
	of (three bit) masked symbols to their two most significant bits
	yields the set $\{(0,0), (t,t),$ $(1,1)\}$, i.e., we count three
	observations. However, the projection to their least significant bit
	yields only the singleton set $\{1\}$, i.e., the observation is
	determined by the masks alone. 
\end{example}

The next proposition shows that counting the symbolic observations
after projecting, as in Example~\ref{ex:projectandcount}, yields an
upper bound for the range of the adversary's view, for {\em any}
valuation of the symbols. We use this effect for static reasoning
about cache-aware memory alignment.

\begin{proposition}\label{prop:quant}
	For every $\absms\in\AbsMS$, every valuation
	$\symbmapping\colon\Sym\rightarrow \Bitvecs$, and every projection
	$\projgen$ mapping vectors to a subset of their components, we have
	$\sizeof{\projgen(\concMS_\symbmapping(\absms))}\le
	\sizeof{\projgen(\absms)}$.
\end{proposition}
\begin{proof}
	This follows from the fact that equality of $\projgen$ on $(s,m)$ and
	$(s',m')$ implies equality of $\projgen$ on $\concMS_\symbmapping(s,m)$
	and $\concMS_\symbmapping(s',m')$, for all $\symbmapping$. To see
	this, assume there is a symbolic bit in $\projgen(s,m)$. Then we have
	$s=s'$, and hence $\symbmapping(s)=\symbmapping(s')$. If there is no
	symbolic bit, the assertion follows immediately.
\end{proof}

\subsection{Update}
We now describe the effect of Boolean and arithmetic operations on
masked symbols. We focus on operations between pairs of masked
symbols; the lifting of those operations to sets is obtained by
performing the operations on all pairs of elements in their
product. The update supports tracking information about known bits
(which are encoded in the mask) and about the arithmetic relationship
between symbols. We explain both cases below.

\subsubsection{Tracking Masks}\label{subsec:trackbits}

Cache-aware code often aligns data to the memory blocks of the
underlying hardware. 

\begin{example}\label{ex:maskop}
	The following code snippet is produced when compiling the
	\verb!align! function from Figure~\ref{fig:CM-openssl} with
	\verb!gcc -O2!.  The register \verb!EAX! contains a pointer to a
	dynamically allocated heap memory location.  {\normalfont
		\begin{lstlisting}[basicstyle=\ttfamily,escapechar=|]
		AND 0xFFFFFFC0, EAX |\label{line:and}|
		ADD 0x40, EAX |\label{line:add}|
		\end{lstlisting}
	}
	The first line ensures that the 6 least significant bits of that
	pointer are set to $0$, thereby aligning it with cache lines of $64$
	bytes. The second line ensures that the resulting pointer points into
	the allocated region while keeping the alignment.
\end{example}
We support different Boolean operations and addition on masked symbols
that enable us to analyze such code. The operations have the form
$(s'',m'')=\absop (s,m),(s',m')$, where the right-hand side denotes
the inputs and the left-hand side the output of the operation. The
operations are defined bit-wise as follows:
\begin{compactitem}
	\item[$\ttop=\ttand$ or $\ttop=\ttor$ :] $m_i''=\ttop m_i, m_i'$, for
	all $i$ such that $m_i,m_i'\in\{0,1\}$, i.e., the abstract $\absop$
	extends the concrete $\ttop$. Whenever $m_i$ or $m_i'$ is absorbing
	(i.e., $1$ for $\ttor$ and $0$ for $\ttand$), we set $m_i''$ to that
	value. In all other cases, we set $m_i''=\top$.
	
	The default is to introduce a fresh symbol for $s''$, unless the
	logical operation acts neutral on all symbolic bits, in which case we can set
	$s''=s$. 
	This happens in two cases: first, if the operands' symbols 
	coincide, i.e. $s=s'$; second, if one operand is constant, i.e. 
	$m'\in\{0,1\}^n$, and $m_i = \top$ implies that $m'_i$ is 
	neutral (i.e., $0$ for $\ttor$ and $1$ for $\ttand$).

	\item[$\ttop=\ttxor$:] $m_i''=\ttxor m_i, m_i'$, for all $i$ such that
	$m_i,m_i'\in\{0,1\}$, i.e., the abstract $\absxor$ extends the
	concrete $\ttxor$.  Whenever the symbols coincide, i.e. $s=s'$, we
	further set $m_i''=0$, for all $i$ with $m_i=m_i'=\top$.  In all
	other cases, we set $m_i''=\top$.
	
	The default is to introduce a fresh symbol for $s''$.  We can avoid
	introducing a fresh symbol and set $s''=s$ in case one operand
	is a constant that acts neutral on each symbolic bit of the other\ifext,
	i.e., if $m'\in\{0,1\}^n$ and if $m_i=\top$ implies
	$m_i'=0$\fi.

	\item[$\ttop=\ttadd$:]
	Starting from $i=0,1,\dots$, and as long as
	$m_i,m'_i\in\{0,1\}$, we compute $m_i''$ according to the standard
	definition of $\ttadd$\ifext\footnote{$\ttadd$ between two bit-vectors $x$
		and $y$ determines the $i$-th bit of the result $r$ as
		$r_i = x_i\oplus y_i \oplus c_i$, where $c_i$ is the carry
		bit. The carry bit is defined by
		$c_{i} = (x_{i-1}\wedge y_{i-1}) \vee (c_{i-1} \wedge x_{i-1})
		\vee (c_{i-1} \wedge y_{i-1})$, with $c_{0} = 0$.}\fi. As soon as we
	reach $i$ with $m_i=\top$ or $m'_i = \top$, we set $m_j''=\top$ for all $j\ge 
	i$.
	
	The default is to use a fresh symbol $s''$, unless
	one operand
	is a constant that acts neutral on the symbolic most-significant bits of 
	the other\ifext,
	i.e., if $m'\in\{0,1\}^n$ and
	for all $j\ge i$,
	$m_j=\top$ implies $m_j'=0$ and $c_j=0$\fi; then we keep the symbol,
	i.e., $s''=s$.
	
	\item[$\ttop=\ttsub$:]
	
	We compute $\ttsub$ similarly to $\ttadd$, 
where the borrow bit takes the role of the carry bit. Here, we use the additional 
    rule that whenever  the symbols coincide, i.e. $s=s'$, we
    further set $m_i''=0$, for all $i$ with $m_i=m_i'=\top$. 
	

%
\end{compactitem}

\begin{example} Consider again Example~\ref{ex:maskop} and assume that
	\verb!EAX! initially has the symbolic value $(\constsym,\top)$. 
	Executing Line~\ref{line:and} yields the masked symbol
	\begin{equation}\label{eq:lsbzero} 
	(\constsym,(\top\cdots\top{}000000))\ ,  
	\end{equation}
	Executing Line~\ref{line:add}, we obtain
	$(\constsym',(\top\cdots\top{}000000))$, for a fresh $\constsym'$. This masked 
	symbol points to the beginning of some (unknown) cache line. In
	contrast, addition of \verb!0x3F!  to \eqref{eq:lsbzero} would yield
	$(\constsym,(\top\cdots\top{}111111))$, for which we can statically
	determine containment in the same cache line as~\eqref{eq:lsbzero}.
\end{example}

\subsubsection{Tracking Offsets}\label{ssec:arith_rel}

Control flow decisions in low-level code often rely on comparisons
between pointers. For analyzing such code with sufficient precision,
we need to keep track of their relative distance.

\begin{example}\label{ex:pointers}
	Consider the function \verb!gather! from Figure~\ref{fig:CM-openssl}.
	When compiled with \verb!gcc -O2!, the loop guard is translated into a
	comparison of pointers. The corresponding pseudocode looks as follows: 
	{\normalfont
		\begin{lstlisting}[mathescape=true]
		y := r + N
		for x := r; x $\neq$ y; x++ do
		*x = buf [k + i * spacing]
		\end{lstlisting}} 
	Here, \verb!r! points to a dynamic memory
	location. The loop terminates whenever pointer \verb!x! reaches
	pointer \verb!y!.
\end{example}
In this section we describe a mechanism that tracks the distance
between masked symbols, and enables the analysis of comparisons, such
as the one in~Example~\ref{ex:pointers}.

\paragraph{Origins and Offsets}

Our mechanism is based on assigning to each masked symbol an {\em
	origin} and an {\em offset} from that origin. The origin of a symbol
is the masked symbol from which it was derived by a sequence addition
operations, and the offset tracks the cumulated effect of these
operations.
\begin{equation*}
\origin\colon  \MSym \rightarrow \MSym\quad\quad
\offset\colon  \MSym\rightarrow \mathbb{N}
\end{equation*}
Initially, $\origin(x)=x$ and $\offset(x)=0$, for all $x\in\MSym$.

For convenience, we also define a partial inverse of $\origin$ and
$\offset$ describing the {\em successor} of an origin at a specific
offset. We formalize this as a function
$\suck\colon\MSym\times \mathbb{N}\rightarrow \MSym\cup\{\nosuccessor\}$ such 
that
$\suck(\origin(x),\offset(x))=x$.

\paragraph{Addition of Offsets}
When performing an addition of a constant to a masked symbol, the
mechanism first checks if there is already a masked symbol with the
required offset. If such a symbol exists, it is reused. If not, the
addition is carried out and memorized.

More precisely, the result of performing the addition $y=\absadd x, \intoff$
of a masked symbol $x\in\MSym$ with a constant $\intoff\in\Bitvecs$ is
computed as follows:
\begin{enumerate}
	\item If $\suck(\origin(x),\offset(x)+\intoff)=x'$ for some masked symbol
	$x'$, then we set $y=x'$.
	\item If $\suck(\origin(x),\offset(x)+\intoff)=\nosuccessor$, then
	we compute $y= \absadd x, \intoff$, as described in
	Section~\ref{subsec:trackbits},
	and update $\origin(y)=\origin(x)$,
	$\offset(y)=\offset(x)+\intoff$, and $\suck(\origin(y),\offset(y))=y$.
\end{enumerate}

Note that we
restrict to the case in which one operand is a constant.
In case {\em both} operands contain symbolic bits, for the result $y$ (obtained 
according to Section~\ref{subsec:trackbits}) we set $\origin(y)=y$ and 
$\offset(y)=0$.


\subsubsection{Tracking Flag Values}\label{subsec:flag_values}
Our analysis is performed at the level of disassembled x86 binary
code, where inferring the status of CPU flags is crucial for
determining the program's control flow. We support limited derivation
of flag values; in particular, we determine the values of
the zero (\texttt{ZF}) and carry flags (\texttt{CF}) as follows.

For the Boolean and addition operations described in
Section~\ref{subsec:trackbits}, we determine the value of flag bits as
follows:
\begin{compactitem}
	\item If at least one masked bit of the result is non-zero, then
	\texttt{ZF = 0}.
	\item If the operation does not affect the (possibly symbolic)
	most-significant bits of the operands, then \texttt{CF = 0}.
\end{compactitem}
For comparison and subtraction operations, we rely on offsets for
tracking their effect on flags. Specifically, for $\abscmp x, y$ and
$\abssub x, y$, with source $x$ and target $y$, we determine the value
of flags as follows:
\begin{compactitem}
	\item If $x = y$, then \texttt{ZF = 1}; \label{itm:endloop}
	\item If $\origin(x)=\origin(y)$ and $\offset(x)\neq\offset(y)$ then \texttt{ZF
		= 0};
\end{compactitem}
In any other case, we assume that all combinations of flag
values are possible.

\begin{example}
	Consider again the code in Example~\ref{ex:pointers}.  Termination
	of the loop is established by an instruction $\ttcmp x,y$, followed
	by a conditional jump in case the zero flag is not set. In our
	analysis, we infer the value of the zero flag by comparing the
	offsets of $x$ and $y$ from their common origin~$r$.
\end{example}

\section{Memory Trace Abstract Domain}\label{sec:tracedom}
In this section, we present the \emph{memory trace domain}, which is a
data structure for representing the set of traces of 
possible memory accesses a
program can make, and for computing the number of observations that
the observers defined in Section~\ref{sec:hierarchy} can make.

\subsection{Representation}
We use a directed acyclic graph (DAG) to compactly represent sets of
traces of memory accesses. This generalizes a data structure that has
been previously used for representing sets of traces of cache hits and
misses~\cite{cacheaudit-tissec15}.

A DAG $\abstr$ from the memory trace domain $\AbsTR$ is a tuple 
$(\Nodes,\Edges,\rt,\Labl,\Rept)$.
The DAG has a set of vertices $\Nodes$ representing program points, with
a root $\rt\in\Nodes$ and a set of edges
$\Edges\subseteq \Nodes\times \Nodes$ representing transitions. We
equip the DAG with a vertex labeling
$\Labl\colon \Nodes\rightarrow\AbsMS$ that attaches to each vertex a set
of masked symbols representing the memory locations that may have been
accessed at this point, together with a repetition count
$\Rept\colon \Nodes\rightarrow\partsof{\mathbb{N}}$ that tracks the
number of times each address has been accessed.
%
During program analysis, we maintain and manipulate a single DAG,
which is why we usually keep $\abstr$ implicit.

\subsection{Concretization}\label{subsec:tracedom-concr}

Each vertex $\node$ in $\abstr$ corresponds to the set of traces of memory
accesses performed by the program from the root up to this point of
the analysis. This correspondence is formally given by a
concretization function $\concTR_\symbmapping$ that is
parameterized by an instantiation
$\symbmapping\colon\Sym\rightarrow\Bitvecs$ of the masked symbols occurring
in the labels (see Section~\ref{sec:statedom}), and is defined~by:
\begin{equation*}
\concTR_\symbmapping(\node) = 
\bigcup_{\node_0\cdots \node_\compsteps }
\{ \evt_0^{r_0}\cdots \evt_\compsteps^{r_\compsteps}\mid \evt_i\in
\concMS_\symbmapping (\Labl(\node_i)), r_i\in
\Rept(\node_i)\},
\end{equation*}
where $\node_0\cdots \node_\compsteps$, 
with $\node_0=\rt$ and $\node_\compsteps=\node$, ranges over all paths from
$\rt$ to $\node$ in
$\abstr$. That is, for each such path, the concretization contains the
adversary's observations (given by the concretizations of the labels
of its vertices) and their number (given by the repetition count).

\subsection{Counting}

We devise a counting procedure that over-approximates the
number of observations different adversaries can make. The key feature
of the procedure is that the bounds it delivers are {\em independent}
of the instantiation of the symbols.

\begin{equation}\label{eq:countview}
\cnt(\node) = \sizeof{\Rept(\node)}\cdot 
\sizeof{\projgen(\Labl(\node))} 
\cdot
\sum\limits_{(\parent,\node)\in \Edges} \cnt(\parent)\ ,
\end{equation}
with $\cnt(\rt)=1$. For the stuttering observers, we replace the
factor $\sizeof{\Rept(\node)}$ from the expression in~\eqref{eq:countview}
by $1$, which captures that those observers cannot distinguish between
repetitions of accesses to the same unit.

\begin{proposition}\label{prop:tracequant}
	For all $\symbmapping\colon\Sym\rightarrow\Bitvecs$
	we have 
	\begin{equation*}
	\sizeof{\view(\concTR_\symbmapping(\node))}\le \cnt(\node)
	\end{equation*}
\end{proposition}
\begin{proof}
	$\cnt(\node)$ recursively sums over all paths from $r$ to $\node$ and weights
	each vertex  with the size of $\projgen$ applied to its label. From
	Proposition~\ref{prop:quant} it follows that this size is larger than
	the the number of concrete observations, regardless of how the symbols
	are instantiated. This yields the assertion.
\end{proof}

\subsection{Update and Join}

The memory trace domain is equipped with functions for update and
join, which govern how sets of traces are extended and merged,
respectively.  

The {\em update} of an abstract element $\abstr$ receives a vertex $\node$
representing a set of traces of memory accesses, and it extends $\node$ by
a new access to a potentially unknown address, represented by a set of
masked symbols $\absms\in\AbsMS$. Technically:
\begin{asparaenum}
	\item If the set of masked symbols is not a repetition (i.e. if
	$\Labl(\node)\neq \absms$), the update function appends a new vertex $\node'$
	to $\node$
	(adding $(\node,\node')$ to $\Edges$), and it sets $\Labl(\node')=\absms$ and
	$\Rept(\node')=\{1\}.$
	\item Otherwise (i.e. if $L(\node)=\absms$), it increments the possible
	repetitions in $\Rept(\node)$ by one.
\end{asparaenum}

The {\em join} for two vertices $\node_1,\node_2$ first checks whether those
vertices have the same parents and the same label, in which case $\node_1$ is
returned, and their repetitions are joined. Otherwise, a new vertex $\node'$
with $\Labl(\node')=\{\epsilon\}$ is generated and edges $(\node_1,\node')$ and
$(\node_2,\node')$ are added to $\Edges$.

\paragraph{Implementation Issues}
To increase precision in counting and compactness of the
representation, we apply the projection $\projgen$ when applying the
update function, and we delay joins until the next update is
performed.  In this way we only maintain the information that is
relevant for the final counting step and can encode accesses to the
same block as stuttering. This is relevant, for example, when an
if-statement fits into a single memory block.


\begin{example}\label{ex:traces}
	Consider the following snippet of x86 assembly, corresponding to a
	conditional branch in libgcrypt~1.5.3:
	
	\begin{minipage}{\linewidth}
		\normalfont\small
		\begin{lstlisting}[basicstyle=\ttfamily,escapechar=|]
		41a90:  mov    0x80(%esp),%eax
		41a97:  test   %eax,%eax
		41a99:  jne    41aa1
		41a9b:  mov    %ebp,%eax
		41a9d:  mov    %edi,%ebp
		41a9f:  mov    %eax,%edi
		41aa1:  sub    $0x1,%edx
		\end{lstlisting}
	\end{minipage}
	Figure~\ref{fig:dag} shows the corresponding DAGs
	for an address-trace observer (Figure~\ref{subfig:dag-addr}) and for a
	block-trace observer (Figure~\ref{subfig:dag-block}) of the
	instruction cache. For both observers, the counting procedure reports
	two traces, i.e., a leak of $1$ bit. For the stuttering block-trace
	observer, however, the counting procedure determines that there is
	only one possible observation.

	\begin{figure}[h]
		\begin{subfigure}[b]{.5\textwidth}
			\includegraphics[scale=.2]{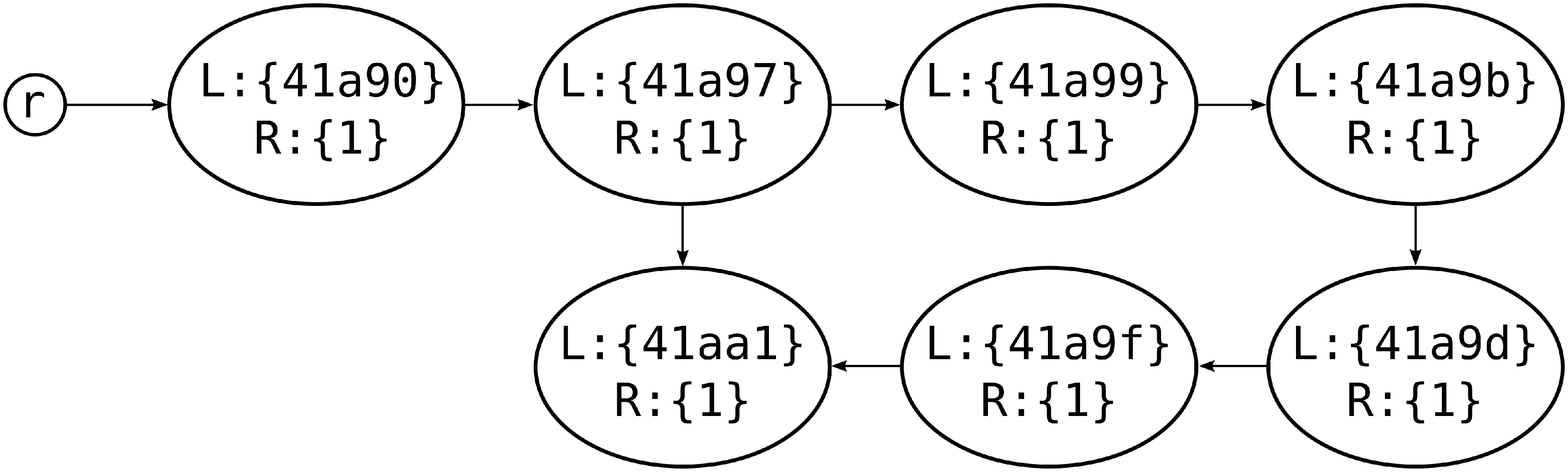}
			\caption{Address-trace observer.}\label{subfig:dag-addr}
		\end{subfigure}
		\begin{subfigure}[b]{.5\textwidth}\centering
			\includegraphics[scale=.2]{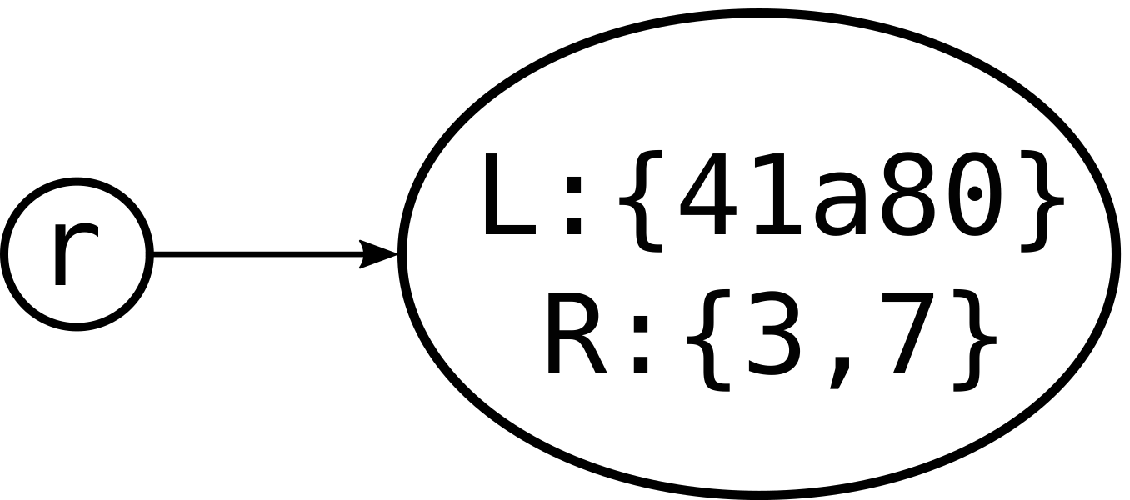}
			\caption{Block-trace observer.}\label{subfig:dag-block}
		\end{subfigure}
		\caption{DAGs that represent the traces corresponding to the assembly code in
			Example~\ref{ex:traces}, for an architecture with cache lines of 
			64 bytes.}\label{fig:dag}
	\end{figure}
\end{example}

%% file: soundness.tex
\vspace{-10pt}
\section{Soundness}\label{sec:soundness}

In this section we establish the correctness of our approach. We split
the argument in two parts. The first part explains how we establish leakage
bounds w.r.t. {\em all} valuations of low variables.  The second part
explains the correctness of the abstract domains introduced in
Section~\ref{sec:statedom} and Section~\ref{sec:tracedom}.

\subsection{Global Soundness and Leakage Bounds}
We formalize the correctness of our approach based on established
notions of local and global soundness~\cite{cousot:cousot77}, which we
slightly extend to cater for the introduction of fresh symbols during
analysis.

For this, we distinguish between the symbols in
$\loSym\subseteq\Sym$ that represent the low initial state and those
in $\Sym\setminus\loSym$ that are introduced during the analysis. A
low initial state in $\StatesInitLo$ is given in terms of a valuation
of low symbols $\statelo\colon \loSym\rightarrow \{0,1\}^n$.
When introducing a fresh symbol $s$ (see Section~\ref{sec:statedom}),
we need to extend the domain of $\statelo$ to include $s$. We denote
by $\Ext(\statelo)$ the set of all functions $\stateloext$ with
$\stateloext\restriction_{\domain{\statelo}}=\statelo$, 
$\domain{\stateloext}\subseteq\Sym$, and
$\range{\stateloext}=\{0,1\}^n$.

With this, we formally define the {\em global soundness} of the
fixpoint $\Abscoll$ of the abstract transition function $\absnext$ as
follows:
\begin{equation}\label{eq:globalsound}
\forall\statelo\in\StatesInitLo\,\exists \stateloext\in\Ext(\statelo):
\collSem_\statelo\subseteq\gamma_{\stateloext}\left( \abs{\collSem} \right)\ .
\end{equation}
Equation~\eqref{eq:globalsound} ensures that for all low initial
states $\statelo$, every trace of the program is included in a
concretization of the symbolic fixpoint, for {\em some} valuation
$\stateloext$ of the symbols that have been introduced during
analysis.

The \emph{existence} of $\stateloext$ is sufficient to
prove our central result, which is a bound for the leakage w.r.t. all
low initial states.

\begin{theorem} Let $\abstr\in\AbsTR$ be the component in $\Abscoll$
  representing memory traces, and $\node\in\abstr$ correspond to the final
  state. Then
$$\Leakage=\max_{\statelo\in\StatesInitLo}\sizeof{\view(\collSemLo)} \leq 
\cnt(\node)$$
\end{theorem}
The statement follows because set inclusion of the fixpoints implies
set inclusion of the projection to memory traces:
$\view(\collSemLo)\subseteq\view(\gamma_\stateloext(\Abscoll))$. The
memory trace of the abstract fixpoint is given by
$\view(\concTR_\stateloext(\node))$, and Proposition~\ref{prop:tracequant}
shows that $\cnt(\node)$ delivers an upper bound, for any $\stateloext$.

\subsection{Local Soundness}

We say that an abstract domain is {\em locally sound} if the abstract
$\absnext$ operator over-approximates the effect of the concrete
$\nextop$ operator (here: in terms of set inclusion). Formally we
require that, for all abstract elements $\absabs$,
\begin{equation}\label{eq:localsoundness}
  \forall\statelo, 
  \exists\stateloext\in\Ext(\statelo): 
  \nextop\left(\conc_\statelo(\absabs)\right)\subseteq\conc_{\stateloext
  }(\absnext(\absabs))\ .
\end{equation}

\paragraph{From Local to Global Soundness}
It is a fundamental result from abstract
interpretation~\cite{cousot:cousot77} that local 
soundness~\eqref{eq:localsoundness} implies
global soundness~\eqref{eq:globalsound}. When the fixpoint is reached in a 
finite number of
steps, this result immediately follows for our modified notions of
soundness, by induction on the number of steps. This is sufficient for
the development in our paper; we leave an investigation of the general
case to future work

\paragraph{Local Soundness in CacheAudit}
It remains to show the local soundness of abstract transfer function
$\absnext$. In our case, this function is inherited from the CacheAudit
static analyzer~\cite{cacheaudit-tissec15}, and it is composed of
several independent and locally sound abstract domains. For details on
how these domains are wired together, refer to~\cite{cacheaudit-tissec15}. 

Here, we focus on proving local soundness conditions for the two
components we developed in this paper, namely the masked symbol and
the memory trace domains.

\paragraph{Masked Symbol Domain}

The following lemma states the local soundness of the Masked Symbol
Domain, following~\eqref{eq:localsoundness}:

\begin{lemma}[Local Soundness of Masked Symbol Domain]
For all operands $\ttop\in\{\ttand,\ttor,\ttxor,\ttadd,\ttsub\}$, we have
 \begin{align*}
   \forall \absms_1,\absms_2,\statelo, 
   \exists\stateloext\in\Ext(\statelo): \\
   \ttop(\concMS_{\statelo}(\absms_1), \concMS_{\statelo}(\absms_2))
   \subseteq \concMS_\stateloext (\absop(\absms_1,\absms_2))
\end{align*}
\end{lemma}
\ifext
\begin{proof}
For the proof, we consider two cases:
\begin{asparaitem}
\item the operation preserves the symbol. Then the abstract update
  coincides with the concrete update, with $\stateloext=\statelo$ and
  $\nextop\left(\concMS_\statelo(\absabs)\right)=\concMS_\statelo(\absnext(\absabs))$. This
  is because the abstract operations are defined such that
  the symbol is only preserved when we can guarantee that the
  operation acts neutral on all symbolic bits.
\item the operation introduces a fresh symbol $s''$. Then we simply
  define $\stateloext(s'')$ such that
  $\stateloext(s'')\maskop m''=\absop(\statelo(s)\maskop m,
  \stateloext(s')\maskop m')$. This is possible because the concrete
  bits in $m''$ coincide with the operation, and the symbolic bits can
  be set as required by $\stateloext$.
\end{asparaitem}
Flag values are correctly approximated as all flag-value combinations
are considered as possible unless the values can be exactly
determined.
\end{proof}
\fi
\paragraph{Memory Trace Domain}
The following lemma states the soundness of the memory trace domain:

\begin{lemma}[Local Soundness of Memory Trace Domain]
 \begin{align*}
  \forall\statelo, 
\exists\stateloext\in\Ext(\statelo):\\ 
\upd\left(\concTR_\statelo(\abstr),\concMS_\stateloext(\absms)\right)\subseteq
\concTR_\stateloext(
    \abs{\upd} (\abstr,\absms))
\end{align*}
\end{lemma}
\begin{proof}
The local soundness of the memory trace domains follows directly 
because the update does not perform any abstraction with respect to
the sets of masked symbols it appends; it just yields a more compact
representation in case of repetitions of the same observation. 
\end{proof}

%% file: experiments.tex
\section{Case Study}\label{sec:casestudy}

This section presents a case study, which leverages the
techniques developed in this paper for the first rigorous analysis of
a number of practical countermeasures against cache side channel attacks on
modular exponentiation algorithms. The countermeasures are
 from releases of the cryptographic 
libraries libgcrypt and OpenSSL from April 2013 to March 2016.  We report on results for
bits of leakage to the adversary models presented in
Section~\ref{sec:hierarchy} (i.e., the logarithm of the maximum number of 
observations the adversaries can make, 
see Section~\ref{sec:quantleak}), due to instruction-cache (I-cache)
accesses and data-cache (D-cache) accesses.\footnote{We also analyzed
  the leakage from accesses to shared instruction- and data-caches;
  for the analyzed instances, the leakage results were consistently
  the maximum of the I-cache and D-cache leakage results.}  As the
adversary models are ordered according to their observational
capabilities, this sheds light into the level of provable security
offered by different protections.


\subsection{Analysis Tool}\label{sec:implement}
We implement the novel abstract domains described in
Sections~\ref{sec:statedom} and~\ref{sec:tracedom} on top of the
CacheAudit open source static
analyzer~\cite{cacheaudit-tissec15}. CacheAudit provides
infrastructure for parsing, control-flow reconstruction, and fixed
point computation. Our novel domains extend the scope of CacheAudit by
providing support for (1) the analysis of dynamically allocated
memory, and for (2) adversaries who can make fine-grained observations
about memory accesses. The source code is publicly 
available\footnote{\url{
http://software.imdea.org/cacheaudit/memory-trace}}.
For all considered instances, our analysis takes between 0 and 4 seconds on 
a 
t1.micro virtual machine instance on Amazon EC2.

\subsection{Target Implementations}\label{sec:target}

The target of our experiments are different side-channel countermeasures 
for modular exponentiation, which we analyse at x86 binary level. 
Our testbed consists 
of C-implementa\-tions of ElGamal decryption~\cite{elgamal1985public} with 
$3072$-bit keys, 
using $6$ different implementations of modular exponentiation, 
which we compile using \texttt{gcc~4.8.4}, on a 32-bit Linux machine.


We use the
ElGamal implementation from the libgcrypt 1.6.3 library, in which 
we replace the source code for modular exponentiation (\texttt{mpi-pow.c}) with 
implementations containing countermeasures from different versions 
of libgcrypt and OpenSSL. 
For libgcrypt, we consider versions 1.5.2 and 1.5.3, which implement 
square-and-multiply modular exponentiation, as well as versions 1.6.1 and 
1.6.3, which implement sliding-window modular exponentiation. Versions 
1.5.2 and 1.6.1 do not implement a dedicated countermeasure against cache 
attacks.
For OpenSSL, we consider versions 1.0.2f and 1.0.2g, which implement 
fixed-window modular exponentiation with two different countermeasures against 
cache attacks. We integrate the countermeasures into the 
libgcrypt~1.6.3-implementation of modular exponentiation.




The current version of CacheAudit supports only a subset of the x86
instruction set, which we extended with instructions required for this case 
study. To bound the
required extensions, we focus our analysis on the regions of the
executables that were targeted by exploits and to which the
corresponding countermeasures were applied, rather than the whole
executables.  As a consequence, the formal statements we derive only
hold for those regions. In particular, we do not analyze the code of
the libgcrypt's multi-precision integer multiplication and modulo
routines, and we specify that the output of the memory allocation
functions (e.g. \texttt{malloc}) is symbolic (see
Section~\ref{sec:statedom}).


\subsection{Square-and-Multiply Modular Exponentiation}\label{sec:squaremult}

The first target of our analysis is modular exponentiation by
square-and-multiply~\cite{menezes96hac}. The algorithm is depicted in
Figure~\ref{fig:square-multiply} and is implemented, e.g., in
libgcrypt version 1.5.2. Line 5 of the algorithm contains a
conditional branch whose condition depends on a bit of the secret
exponent. An attacker who can observe the victim's accesses to
instruction or data caches may learn which branch was taken and
identify the value of the exponent bit. This weakness has been shown
to be vulnerable to key-recovery attacks based on
prime+probe~\cite{LiuYGHL15, ZhangJRR12} and
flush+reload~\cite{YaromF14}.

In response to these attacks, libgcrypt 1.5.3 implements a
countermeasure that makes sure that the squaring operation is always
performed, see Figure~\ref{fig:square-always-multiply} for the
pseudocode.  It is noticeable that this implementation still contains
a conditional branch that depends on the bits of the exponent in 
Lines~7--8, namely the copy operation that selects the outcome of both
multiplication operations. However, this has been considered a minor
problem because the branch is small and is expected to fit into the
same cache line as preceding and following code, or to be always
loaded in cache due to speculative execution~\cite{YaromF14}. In the
following, we apply the techniques developed in this paper to analyze
whether the expectations on memory layout are met.\footnote{Note that
  we analyze the branch in Lines 7--8 for one iteration; in the
  following iteration the adversary may learn the information by
  analyzing which memory address is accessed in Line 3 and 4.}

\begin{figure}[h]\centering
\begin{minipage}{.7\linewidth}
\begin{lstlisting}[mathescape]
r := 1
for i := |e| - 1 downto 0 do
    r := mpi sqr(r)
    r := mpi mod(r, m)
    if $e_i$ = 1 then
        r := mpi_mul(b, r)
        r := mpi_mod(r, m)
    return r
\end{lstlisting}
\end{minipage}
\vspace{-1em}
\caption{Square-and-multiply modular exponentiation  
}\label{fig:square-multiply}

\begin{minipage}{.7\linewidth}
\begin{lstlisting}[mathescape]
r := 1
for i := |e| - 1 downto 0 do
    r := mpi_sqr(r)
    r := mpi_mod(r, m)
    tmp := mpi_mul(b, r)
    tmp := mpi_mod(tmp , m)
    if $e_i$ = 1 then
        r := tmp
    return r
\end{lstlisting}
\end{minipage}
\vspace{-1em}
\caption{Square-and-always-multiply  
exponentiation}\label{fig:square-always-multiply}
\end{figure}
\begin{figure}[h]\small
\centering
\begin{subfigure}{.45\textwidth}
\centering
\begin{tabular}[h]{r|ccc}
 Observer  & \addrtrace & \blocktrace & \weaktrace\\\hline
 I-Cache  & $1$ bit   & $1$ bit & $1$ bit  \\
 D-Cache  & $1$ bit   & $1$ bit & $1$ bit  \\
\end{tabular}
\vspace{-.5em}
\caption{Square-and-multiply from 
libgcrypt~1.5.2}\label{fig:res-1.5.2}
\end{subfigure}
\quad\quad
\begin{subfigure}{.45\textwidth}
\centering
\begin{tabular}[h]{r|ccc }
 Observer  & \addrtrace & \blocktrace   & \weaktrace \\\hline
 I-Cache   & $1$ bit & $1$ bit & $0$ bit \\
 D-Cache   & $0$ bit & $0$ bit & $0$ bit \\
\end{tabular}
\vspace{-.5em}
\caption{Square-and-{\em always}-multiply from
  libgcrypt~1.5.3}\label{fig:res-1.5.3-O2}
\end{subfigure}
\vspace{-1em}
\caption{Leakage of modular exponentiation algorithms to observers of
  instruction and data caches, with cache line size of 64 bytes and
  compiler optimization level~\texttt{-O2.}}\label{fig:sqandm}
\end{figure}

\begin{figure}[h]\small
\centering
 \begin{tabular}[h]{r|ccc }
 Observer  & \addrtrace & \blocktrace   & \weaktrace \\\hline
 I-Cache   & $1$ bit & $1$ bit & $1$ bit \\
 D-Cache   & $1$ bit & $1$ bit & $1$ bit \\
 \end{tabular}
 \vspace{-1em}
\caption{Leakage of square-and-always-multiply from libgcrypt~1.5.3,
  with cache line size of 32 bytes and compiler optimization level 
\texttt{-O0}.}\label{fig:res-1.5.3-O0}
\end{figure}

\begin{figure}[h]
\centering
\begin{subfigure}[b]{.45\textwidth}\centering
\includegraphics[width=.8\textwidth]{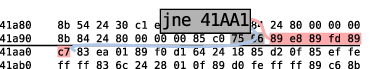}
\caption{Compiled with the default gcc optimization level -O2. Regardless 
whether the jump is taken or not, first block \texttt{41a80} is accessed, 
followed by block \texttt{41aa0}. This results in a 0-bit 
\weaktrace-leak.}\label{subfig:layout-1.5.3-O2}
\end{subfigure}
\begin{subfigure}[b]{.45\textwidth}\centering
\includegraphics[width=.8\textwidth]{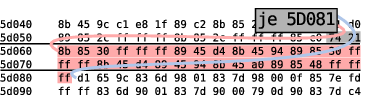}
\caption{Compiled with gcc optimization level -O0. The memory block 
\texttt{5d060} is only accessed when the jump is taken. This results in a 1-bit 
\weaktrace-leak.}\label{subfig:layout-1.5.3-O0}
\end{subfigure}
\caption{Layout of libgcrypt~1.5.3 executables with 32-byte memory
  blocks (black lines denote block boundaries). The highlighted code
  corresponds to the conditional branching in lines 7--8 in
  Figure~\ref{fig:square-always-multiply}. The red region corresponds
  to the executed instructions in the if-branch. The blue curve points
  to the jump target, where the jump is taken if the if-condition does
  not hold.}\label{fig:layout-1.5.3}
\end{figure}

\paragraph{Results} 
The results of our analysis are given in Figure~\ref{fig:sqandm} and 
Figure~\ref{fig:res-1.5.3-O0}.
\begin{asparaitem} 
\item Our analysis identifies a 1-bit data cache leak in
  square-and-multiply exponentiation (line 2 in
  Figure~\ref{fig:res-1.5.2}), due to memory accesses in the
  conditional branch in that implementation. Our analysis confirms
  that this data cache leak is closed by square-and-always-multiply
  (line 2 in Figure~\ref{fig:res-1.5.3-O2}).
\item Line 1 of Figures ~\ref{fig:res-1.5.2} and
  Figure~\ref{fig:res-1.5.3-O2} show that both implementations leak
  through instruction cache to powerful adversaries who can see each
  access to the instruction cache. However, for weaker, stuttering block-trace (\weaktrace)
  observers that cannot distinguish between repeated accesses to a
  block, square-and-always-multiply does {\em not} leak, confirming the
  intuition that the conditional copy operation is indeed less problematic
  than the conditional multiplication. 
\item The comparison between Figure~\ref{fig:res-1.5.3-O2} and 
Figure~\ref{fig:res-1.5.3-O0} demonstrates that the effectiveness of 
countermeasures can depend on  details such as cache line size and compilation
strategy. This is illustrated in Figure~\ref{fig:layout-1.5.3}, which 
shows that more aggressive
compilation leads to more compact code that fits into single cache
lines.  The same effect is observable for data caches, where more
aggressive compilation avoids data cache accesses altogether.

\end{asparaitem}

\begin{figure}[h]\centering
\begin{minipage}{.7\linewidth}
\begin{lstlisting}[mathescape]
  if e0 == 0 then
      base_u := bp
      base_u_size := bsize
  else
      base_u := b_2i3[e0 - 1]
      base_u_size := b_2i3size[e0 - 1]
\end{lstlisting}
\end{minipage}
 \vspace{-1em}
\caption{Table lookup from libgcrypt~1.6.1. %
  Variable \texttt{e0} represents the window, right-shifted by $1$. The lookup
  returns a pointer to the first limb of the multi-precision integer
  in \texttt{base\_u}, and the number of limbs in \texttt{base\_u\_size}. The 
first branch
  deals with powers of $1$ by returning pointers to the base. 
}\label{fig:1.6.1-src}
\end{figure}

\subsection{Windowed Modular Exponentiation}\label{sec:windowed-exp}

In this section we analyze windowed algorithms for modular
exponentiation~\cite{menezes96hac}. These algorithms differ from algorithms 
based on
square-and-multiply in that they process multiple exponent bits in one
shot. For this they commonly rely on tables filled with pre-computed
powers of the base. For example, for moduli of $3072$ bits,
libgcrypt 1.6.1 pre-computes $7$
multi-precision integers and handles the power~$1$ in a branch, see
Figure~\ref{fig:1.6.1-src}.
Each pre-computed value requires $384$ bytes
of storage, which amounts to $6$ or $7$ memory blocks in architectures
with cache lines of $64$ bytes. Key-dependent accesses to those tables
can be exploited for mounting cache side channel
attacks~\cite{LiuYGHL15}.


We consider three countermeasures, which are commonly deployed to defend 
against this
vulnerability. They have in common that they all {\em copy} the
table entries instead of returning a pointer to the entry.
\begin{figure}[h]\centering
\begin{minipage}{.7\linewidth}
\begin{lstlisting}[mathescape,escapechar=~]
// Retrieves r from p[k]
secure_retrieve ( r , p , k):
    for i := 0 to n - 1 do
      for j := 0 to N - 1 do
        v := p[i][j]
        s := (i == k)
        r[j] := r[j] ^ ((0 - s) & ( r[j] ^ v))~\label{ln:cp}~
\end{lstlisting}
\end{minipage}
 \vspace{-1em}
\caption{A defensive routine for array lookup with a constant sequence of
  memory accesses, as implemented in libgcrypt
  1.6.3.}\label{fig:CM-1.6.3}
\end{figure}
\begin{asparaitem}
\item The first countermeasure ensures that in the copy process, a constant 
sequence of memory locations is accessed, see 
Figure~\ref{fig:CM-1.6.3}
  for pseudocode. The expression on line \ref{ln:cp} ensures that
  only the $k$-th pre-computed value is actually copied to
  \texttt{r}. This countermeasure is implemented, e.g. in NaCl and libgcrypt
  1.6.3.
\item The second countermeasure stores pre-computed values in such a
  way that the $i$-th byte of all pre-computed values resides in the
  same memory block. This ensures that when the pre-computed values are
  retrieved, a constant sequence of memory blocks will be
  accessed. This so-called scatter/gather technique is described in
  detail in Section~\ref{sec:illustration}, with code in
  Figure~\ref{fig:CM-openssl}, and is deployed, e.g. in OpenSSL~1.0.2f.
\item The third countermeasure is a variation of scatter/gather, and ensures 
that the gather-procedure performs a constant sequence of memory accesses (see 
Figure~\ref{fig:CM-gather-const}). 
This countermeasure was recently introduced in OpenSSL~1.0.2g, as a response to 
the CacheBleed attack, where the adversary can use
\emph{cache-bank conflicts} to make finer-grained observations and
recover the pre-computed values despite scatter/gather. For example,
the pre-computed values in Figure~\ref{fig:scatter-gather-layout} will
be distributed to different cache banks as shown in
Figure~\ref{fig:precomp-banks}, and cache-bank adversaries can
distinguish between accesses to $p_0,\dots,p_3$ and $p_4,\dots,p_7$.
\end{asparaitem}

\begin{figure}[h]
 \centering
 \begin{minipage}{.7\linewidth}
 \begin{lstlisting}[mathescape]
defensive_gather( r, buf, k ):
    for i := 0 to N-1 do
        r[i] := 0
        for j:= 0 to spacing - 1 do
            v := buf[j + i*spacing]
            s := (k == j)
            r[i] := r[i] | (v & (0 - s))
\end{lstlisting}

\end{minipage}
 \vspace{-1em}
\caption{A defensive implementation of gather (compare to 
Figure~\ref{fig:CM-openssl}) from 
OpenSSL~1.0.2g.}\label{fig:CM-gather-const}
\end{figure}

\begin{figure}[h]
	\centering
	\includegraphics[width=.4\textwidth]{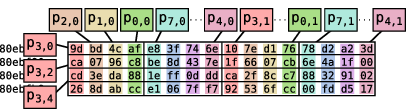}
	\vspace{-1em}
	\caption{Layout of pre-computed values in cache banks, for a platform with 
		$16$ banks of $4$ bytes. The shown data fits into one memory block, and 
		the cells of the grid represent the cache banks.}\label{fig:precomp-banks}
\end{figure}

\paragraph{Results} 
Our analysis of the different versions of the table lookup yields the
following results\footnote{We note that sliding-window exponentiation
  exhibits further control-flow vulnerabilities, some of which we also
  analyze. To avoid redundancy with Section~\ref{sec:squaremult}, we
  focus the presentation of our results on the lookup-table
  management.}:
  
  \begin{figure}[h]\small
\begin{subfigure}{.5\textwidth}
\centering
\begin{tabular}[h]{r|ccc}
 Observer  & \addrtrace & \blocktrace & \weaktrace \\\hline
 I-Cache    &  $1$ bit & $1$ bit & $1$ bit \\
 D-Cache   & $5.6$ bit & $2.3$ bit & $2.3$ bit \\
\end{tabular}
 \vspace{-.5em}
\caption{Leakage of 
secret-dependent table lookup in the modular exponentiation 
implementation from libgcrypt~1.6.1.}\label{fig:res-1.6.1-O2}
\end{subfigure}
\begin{subfigure}{.5\textwidth}
\centering
\begin{tabular}[h]{r|ccc}
 Observer  & \addrtrace & \blocktrace   & \weaktrace \\\hline
 I-Cache   & $0$ bit & $0$ bit & $0$ bit \\
 D-Cache   & $0$ bit & $0$ bit & $0$ bit \\
\end{tabular}
 \vspace{-.5em}
\caption{Leakage in the patch from 
libgcrypt~1.6.3.}\label{fig:res-1.6.3}
\end{subfigure}
\begin{subfigure}{.5\textwidth}
\centering
\begin{tabular}[h]{r|ccc }
 Observer  & \addrtrace & \blocktrace   & \weaktrace \\\hline
 I-Cache   & $0$ bit & $0$ bit & $0$ bit \\
 D-Cache   & $1152$  bit & $0$ bit & $0$ bit \\
\end{tabular}
\vspace{-.5em}
\caption{Leakage in the scatter/gather technique, 
applied to libgcrypt~1.6.1.}\label{fig:res-1.6.1-CM}
\end{subfigure}
\begin{subfigure}{.5\textwidth}
\centering
\begin{tabular}[h]{r|ccc }
 Observer  & \addrtrace & \blocktrace   & \weaktrace \\\hline
 I-Cache   & $0$ bit & $0$ bit & $0$ bit \\
 D-Cache   & $0$  bit & $0$ bit & $0$ bit \\
\end{tabular}
\vspace{-.5em}
\caption{Leakage in the defensive gather 
technique from OpenSSL~1.0.2g, 
applied to libgcrypt~1.6.1.}\label{fig:res-const-gather}
\end{subfigure}
 \vspace{-1em}
\caption{Instruction and data cache leaks of different
  table lookup implementations. Note that the leakage in
  Figure~\ref{fig:res-1.6.1-O2} accounts for copying a pointer, whereas
  the leakage in Figure~\ref{fig:res-1.6.3} and~\ref{fig:res-1.6.1-CM}
  refers to copying multi-precision integers.}\label{fig:lookup}
\end{figure}

\begin{figure}[h]
\centering
\begin{subfigure}[b]{.45\textwidth}\centering
\includegraphics[width=.8\textwidth]{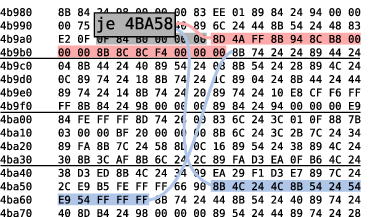}
\caption{Compiled with the default gcc optimization level -O2. If the jump is 
taken, first block \texttt{4b980}, followed by block \texttt{4ba40}, followed 
by \texttt{4b980} again. If the branch is not taken, only block \texttt{4b980} 
is accessed. }\label{subfig:1.6.1-O2}
\end{subfigure}
\begin{subfigure}[b]{.45\textwidth}\centering
\includegraphics[width=.8\textwidth]{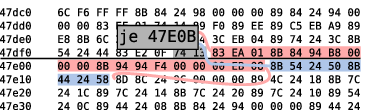}
\caption{Compiled with gcc optimization level -O1. Regardless whether the jump 
is taken or not, first 
block \texttt{47dc0} is accessed, followed by block 
\texttt{47e00}.}\label{subfig:1.6.1-O1}
\end{subfigure}
\caption{Layout of executables using libgcrypt~1.6.1. The highlighted code 
corresponds to a conditional branch (blue: if-branch, red: else-branch).
Curves represent jump targets.}
\end{figure}\label{fig:1.6.1-compiler}

\begin{asparaitem}
\item Figure~\ref{fig:res-1.6.1-O2} shows the results of the analysis
  of the unprotected table lookup in Figure~\ref{fig:1.6.1-src}. The
  leakage of one bit for most adversaries is explained by the fact
  that they can observe which branch is taken. The 
layout of the conditional branch is demonstrated in 
Figure~\ref{subfig:1.6.1-O2}; lowering the optimization level results in a 
different layout (see Figure~\ref{subfig:1.6.1-O1}), and in this case our 
analysis shows that the I-Cache \weaktrace-leak is eliminated.
\item Figure~\ref{fig:res-1.6.1-O2} also shows that more powerful
  adversaries that can see the exact address can learn $\log_2 7=2.8$
  bits per access. The static analysis is not precise enough to
  determine that the lookups are correlated, hence it reports that at
  most $5.6$ bits are leaked.  
\item Figure~\ref{fig:res-1.6.3} shows that the defensive copying
  strategy from libgcrypt~1.6.3 (see
  Figure~\ref{fig:CM-1.6.3}) eliminates all leakage to the cache.
\item Figure~\ref{fig:res-1.6.1-CM} shows that the scatter/gather
  copying-strategy eliminates leakage for any adversary that can
  observe memory accesses at the granularity of memory blocks, and
  this constitutes the first proof of security of this
  countermeasure. For adversaries that can see the full address-trace,
  our analysis reports a $3$-bit leakage for each memory access, which
  is again accumulated over correlated lookups because of imprecisions
  in the static analysis. 
\item Our analysis is able to detect the leak leading to the CacheBleed
attack~\cite{yarom16cachebleed} against 
scatter/gather. The leak 
is visible when comparing the results 
of 
the analysis in Figure~\ref{fig:res-1.6.1-CM} with respect to
address-trace and block-trace adversaries, however, its severity may
be over-estimated due to the powerful address-trace observer.
For a more accurate analysis of this threat,
we repeat the analysis for the bank-trace D-cache
observer. The analysis results in 384-bit leak, which corresponds to
one bit leak per memory access, accumulated for each accessed byte due
to analysis imprecision (see above). The one-bit leak in the $i$-th memory
access is explained by the ability of this observer to distinguish
between the two banks within which the $i$-th byte of all pre-computed 
values fall.
  \item Figure~\ref{fig:res-const-gather} shows that defensive gather 
from 
OpenSSL 1.0.2g (see Figure~\ref{fig:CM-gather-const}) eliminates all leakage to 
cache.
\end{asparaitem}

%
%
%

\input{discussion}

\ifext

\begin{figure*}[h]\centering\footnotesize
\begin{subfigure}{.6\textwidth}
\centering
\begin{tabular}{r|cc|cccc}
algorithm & \multicolumn{2}{c|}{square and multiply} & 
\multicolumn{4}{c}{sliding window} \\\hline
countermeasure 
 &\multirow{2}{*}{no CM} & always &\multirow{2}{*}{no CM} & scatter/& 
access all& defensive\\
(CM)&       & multiply &       & gather     & 
 bytes   & gather       \\\hline
\multirow{2}{*}{implementation} & libgcrypt & libgcrypt & libgcrypt & 
openssl & libgcrypt & openssl \\
& 1.5.2 & 1.5.3 & 1.6.1 & 1.0.2f & 
1.6.3 & 1.0.2g \\\hline
instructions $\times 10^6$ & $90.32$ & $120.62$ & 
$73.99$ & $74.21$ &  $74.61$ & $75.29$ \\
cycles $\times 10^6$ 
& $75.58$ & $100.73$ & $61.58$ & $61.65$ & $62.20$ & $62.28$\\
\end{tabular}
 \vspace{-.5em}
\caption{Different versions of 
modular exponentiation.  
}\label{fig:performance}
\end{subfigure}\qquad
\begin{subfigure}{.33\textwidth}\centering
\begin{tabular}{r|ccc}
algorithm & 
\multicolumn{3}{c}{sliding window} \\\hline
\multirow{2}{*}{countermeasure} 
 & scatter/& access all & defensive\\
 & gather  &  bytes     & gather       \\\hline
\multirow{2}{*}{implementation} & 
openssl & libgcrypt & openssl \\
             & 1.0.2f & 1.6.3  & 1.0.2g \\\hline
instructions & $2991$ & $8618$ & $13040$
\\
cycles       & $859$ & $3073$ & $5579$ \\

\end{tabular}
 \vspace{-.5em}
\caption{Only multi-precision-integer 
retrieval step.}\label{fig:performance-2}
\end{subfigure}
 \vspace{-1em}
\caption{Performance measurements for libgcrypt~1.6.3 ElGamal decryption, for 
3072-bit keys.}\label{fig:performance-all}
\end{figure*}


\subsection{Performance of Countermeasures}
We conclude the case study by considering the effect of the different 
countermeasures on the performance of 
modular exponentiation. For the target implementations (see Section~\ref{sec:target}),
we measure performance as the clock count (through the \texttt{rdtsc} 
instruction),
 as well as the number of performed instructions (through the 
\texttt{PAPI} library~\cite{mucci1999papi}), for performing exponentiations, 
for a sample of random bases and exponents. We make 100,000 measurements with 
an Intel~Q9550 CPU.

 Figure~\ref{fig:performance} summarizes our measurements. The results show 
that the applied countermeasure for square and multiply causes a significant 
slow-down of the exponentiation. A smaller slow-down is observed with 
sliding-window countermeasures as well; this slow-down is demonstrated in 
Figure~\ref{fig:performance-2}, which 
shows the performance of the retrieval of pre-computed values, with 
different countermeasures applied.


\fi

%% file: discussion.tex
\subsection{Discussion}

A number of comments are in order when interpreting the bounds
delivered by our analysis.

\paragraph{Use of Upper Bounds}
The results we obtain are upper bounds on the leaked information.
Results of zero leakage present a proof of the
absence of leaks.
Positive leakage bounds, however, are not necessarily tight and
do not correspond to proofs
of the presence of leaks.
The reason for this is that the amount of
leaked information may be over-estimated due to imprecision of the
static analysis, as is the case with the D-Cache leak shown on
Figure~\ref{fig:res-1.6.1-CM}.

\paragraph{Use of Imperfect Models}
The guarantees we deliver are only valid to the extent to which the
models used accurately capture the aspects of the execution platform
relevant to known attacks.  A recent empirical study of OS-level side
channels on different platforms~\cite{Cock:2014} shows that advanced
microarchitectural features may interfere with the cache, which can
render countermeasures ineffective.

\paragraph{Alternative Attack Targets}
In our analysis, we assume that heap addresses returned by malloc are
low values. For analyzing scenarios in which the heap addresses
themselves are the target of cache attacks (e.g., when the goal is to
reduce the entropy of ASLR~\cite{hund2013practical}), heap addresses
would need to be modeled as high data.

%% file: related.tex
\section{Related Work}\label{sec:related}
We begin by discussing approaches that tackle related goals, before we
discuss approaches that rely on similar techniques.

\paragraph{Transforming out Timing Leaks}
Agat proposes a program transformation for removing control-flow
timing leaks by equalizing branches of conditionals with secret
guards~\cite{inp:Agat2000a}, which he complements with an informal
discussion of the effect of instruction and data caches in Java byte
code~\cite{tr:agat2000}. 
Molnar et al.~\cite{molnar2006program} propose a program
transformation that eliminates control-flow timing leaks, together
with a static check for the resulting x86 executables. Coppens et
al.~\cite{CoppensVBS09} propose a similar transformation and evaluate
its practicality. The definitions in Section~\ref{sec:security}
encompass the adversary model of~\cite{molnar2006program}, but also
weaker ones; they could be used as a target for program
transformations that allow for limited forms for secret-dependent
behavior.

\paragraph{Constant-time Software}

Constant-time code defeats timing attacks by ensuring that control
flow, memory accesses, and execution time of individual instructions
do not depend on secret data. Constant-time code is the current gold
standard for defensive implementations of cryptographic
algorithms~\cite{BernsteinLS12}.

A number of program analyses support verifying constant-time
implementations. Almeida et al.~\cite{AlmeidaBPV13} develop an
approach based on self-composition that checks absence of timing leaks
in C-code; Almeida et al.~\cite{almeida2016verifying} provide a tool
chain for verifying constaint-time properties of LLVM IR code. Similar
to the dynamic analysis by Langley~\cite{langley-valgrind}, our
approach targets executable code, thereby avoiding potential leaks
introduced by the compiler~\cite{kaufmann2016constant}. Moreover, it
supports verification of more permissive interactions between software
and hardware -- at the price of stronger assumptions about the
underlying hardware platform.


\paragraph{Quantitative Information Flow Analysis}
Technically, our work draws on methods from quantitative
information-flow analysis (QIF)~\cite{ClarkHM07}, where the automation
by reduction to counting problems appears
in~\cite{backeskoepfrybal09,newsome09}, and has subsequently been
refined in several
dimensions~\cite{HeusserM10,koepfrybal10,Klebanov12a,malacaria16}.

Specifically, our work builds on
CacheAudit~\cite{cacheaudit-tissec15}, a tool for the static
quantification of cache side channels in x86 executables. The
techniques developed in this paper extend CacheAudit with support for
precise reasoning about dynamically allocated memory, and a rich set of
novel adversary models. Together, this enables the first formal
analysis of widely deployed countermeasures, such as scatter/gather.

\paragraph{Abstract Interpretation}
We rely on basic notions from abstract
interpretation~\cite{cousot:cousot77} for establishing the soundness
of our analysis. However, the connections run deeper: For example, the
observers we define (including the stuttering
variants~\cite{gm05formats}) can be seen as abstractions in the
classic sense, which enables composition of their views in algebraic
ways~\cite{CousotC79}.  Abstract interpretation has also been used for
analyzing information flow properties~\cite{GM04,AssafNSTT17}. Reuse
of the machinery developed in these papers could help streamline our
reasoning. We leave a thorough exploration of this connection to
future work.